\documentclass[journal,onecolumn,12pt,draft]{IEEEtranTCOM}
\usepackage{amsfonts,color,morefloats}
\usepackage{booktabs}
\usepackage{amssymb,amsmath,latexsym,amsthm}
\usepackage{graphicx}
\usepackage{makecell}
\usepackage{hyperref}
\usepackage{multirow}
\usepackage{array}

\newtheorem{theorem}{Theorem}[section]
\newtheorem{lemma}[theorem]{Lemma}

\newtheorem{example}[theorem]{Example}

\newcommand{\C}{\mathcal{C}}

\renewcommand\thesection{\arabic{section}}

\begin{document}
	
	\title{Four infinite families of ternary cyclic codes with a square-root-like lower bound\thanks{C. Ding's research was supported by The Hong Kong Research Grants Council, Proj. No. $16302121$.
C. Li's research was supported by the National
Natural Science Foundation of China (12071138) and Shanghai Natural Science Foundation (22ZR1419600).
Z. Sun's research was supported by The National Natural Science Foundation of China under Grant Number 62002093.
}}	
\author{Tingfang Chen\thanks{T. Chen is with the Department of Computer Science
                           and Engineering, The Hong Kong University of Science and Technology,
Clear Water Bay, Kowloon, Hong Kong, China. Email: tchenba@connect.ust.hk}, \and
Cunsheng Ding\thanks{C. Ding is with the Department of Computer Science
                           and Engineering, The Hong Kong University of Science and Technology,
Clear Water Bay, Kowloon, Hong Kong, China. Email: cding@ust.hk}, \and
Chengju Li\thanks{C. Li is with the Shanghai Key Laboratory of Trustworthy Computing, East China Normal University,
Shanghai, 200062, China; and also with the National Mobile Communications Research Laboratory, Southeast University, China.  Email: cjli@sei.ecnu.edu.cn}, \and
Zhonghua Sun\thanks{Z. Sun is with the School of Mathematics, Hefei University of Technology, Hefei, 230601, Anhui, China. Email:  sunzhonghuas@163.com}
}
	\date{\today}
	\maketitle
	
\begin{abstract}
Cyclic codes are an interesting type of linear codes
and have wide applications in communication and storage systems
due to their efficient encoding and decoding algorithms. Inspired by the recent work on binary cyclic codes published in IEEE Trans. Inf. Theory, vol. 68, no. 12,
pp. 7842-7849, 2022, and the arXiv paper  arXiv:2301.06446,
the objectives of this paper are the construction and analyses of
four infinite families of ternary cyclic codes with length $n=3^m-1$ for odd $m$ and dimension $k \in \{n/2, (n + 2)/2\}$ whose minimum distances have a square-root-like lower bound. Their  duals have parameters $[n, k^\perp, d^\perp]$,
where $k^\perp \in \{n/2, (n- 2)/2\}$ and $d^\perp$ also has a square-root-like lower bound. These families of codes and their
duals contain distance-optimal cyclic codes.
	\end{abstract}
	


\section{Introduction} \label{sec-intro}

In this paper, we assume that the reader is familiar with the basics of linear codes and cyclic codes over finite fields
\cite{Charpin,HP03} and do not introduce them here. 
Let $\C$ be an $[n, k, d]$ linear
code over $\mathbb{F}_q$. $\C$ is said to be \textit{distance-optimal}  if there does not exist an $[n, k, d' \geq d+1]$  linear code over $\mathbb{F}_q$. A code is called an optimal code if it is
 distance-optimal or meets a bound for linear codes. If a lower bound on the minimum weight $d$ of $\C$
is close to $\sqrt{n}$, we then call it a
\emph{square-root-like lower bound}.

Cyclic codes  have wide applications in communication and storage systems
due to their efficient encoding and decoding algorithms \cite{HP03}. They are also important in theory, as they are closely
related to algebra, algebraic geometry, algebraic function fields, algebraic number theory, association schemes,
combinatorics, finite fields, finite geometry, graph theory, and group theory. There are a lot of references on 
cyclic codes over finite fields. However, no infinite family of ternary cyclic codes with parameters 
$[n, k \in \{n/2, (n\pm 2)/2\}, d\geq \sqrt{n}]$ is seen in the literature. It is very difficult to design 
such an infinite family of cyclic codes, as the dimension is about half of the length of each code in the family. These are the motivations of this paper. 

Recently, several infinite families
of binary cyclic codes with parameters $[n, (n \pm 1)/2, d]$ were constructed, where the minimum distance $d$
has a lower bound which is quite close to $\sqrt{n}$ \cite{TD22,LiuLiDing23,LLQ23}. Inspired by the recent results
in \cite{TD22,LiuLiDing23,LLQ23}, in this paper we construct four infinite families of ternary cyclic codes with parameters
$[n, k, d]$, where $k \in \{n/2, (n \pm 2)/2\}$ and the minimum distance $d$
has a lower bound which is quite close to $\sqrt{n}$.
Their duals have parameters $[n, k^\perp, d^\perp]$, where $k^\perp \in \{n/2, (n \pm 2)/2\}$ and the minimum distance $d^\perp$ also
has a lower bound which is quite close to $\sqrt{n}$.
These four families of codes and their duals contain distance-optimal cyclic codes.
The authors are not aware of such an infinite family of ternary cyclic codes
in the literature.

The rest of this paper is organized as follows. In Section \ref{sec-construction}, we give a general construction of ternary cyclic codes.
In Section \ref{sec-dimension}, we settle the dimensions of the four infinite families of ternary cyclic codes.
In Sections \ref{sec-four1}-\ref{sec-four4}, we develop lower bounds on their minimum distances.
In Section \ref{sec-dual}, we investigate the parameters of their dual codes.
In Section \ref{sec-conclude}, we conclude this paper and make some remarks.
	
\section{A general construction of ternary cyclic codes}\label{sec-construction}

Let  $\mathbb{Z}_n =\{0,1,2,\ldots,n-1\}$ be the ring of integers modulo $n$.
For any integer $b$, let $b \bmod n$ denote the unique $b_0 \in \mathbb{Z}_n$ such that $b \equiv b_0 \pmod n$.
For any $s \in \mathbb{Z}_n$, the $3$-cyclotomic coset of $s$ modulo $n$ is defined by
	$$C_s^{(3,n)}=\{s,s3,s3^2,\ldots,s3^{l_s-1}\}\bmod n \subseteq \mathbb{Z}_n,$$
	where $l_s$ is the smallest positive integer such that $s \equiv s3^{l_s}\pmod  n$.
For an integer $i$ with $0 \le i \le 3^m-1$, let $$i=i_{m-1}3^{m-1} + i_{m-2}3^{m-2} + \cdots + i_1 3+i_0$$ be the $3$-adic expansion of $i$, where
	$i_j \in \{0, 1, 2\}$ for $0 \le j \le m-1$.
For any $i$ with $ 0 \leq i \leq n-1 $, define $w_{3}(i) = \sum\limits_{j=0}^{m-1}i_{j}$, which is called the
$3$-weight of $i$.

From now on, we fix $n=3^m-1$ and let $\alpha$ denote a primitive element of $\mathbb{F}_{3^m}$, where $m \geq 2$.
Define a polynomial
\begin{equation} \label{equ-gi1i2m}
	g_{(i_1,i_2,m)}(x) =  \prod\limits_{
		\substack{1 \leq j \leq n-1\\
			w_{3}(j) \equiv i_1 \text{ or } i_2 \pmod{4}}}
	(x-\alpha^{j}),
\end{equation}
where $i_1$ and $i_2$ are a pair of distinct elements in the set $\{0,1,2,3\}$.
It is easily seen that $g_{(i_1,i_2,m)}(x) \in \mathbb{F}_3[x]$, as $\omega_3(j)$ is a constant on each $3$-cyclotomic
coset $C_s^{(3,n)}$. Let $\mathcal C_{(i_1,i_2,m)}$ denote the ternary cyclic code of length $n=3^m-1$ with generator polynomial $g_{(i_1,i_2,m)}(x)$.
 Denote
   $$T_{(i_1,i_2,m)} = \{ 1 \leq j \leq n-1: \ \ w_{3}(j) \equiv i_1 \text{ or } i_2 \pmod{4}\}.$$
   It is  clear that $T_{(i_1,i_2,m)}$ is the defining set of $\mathcal C_{(i_1,i_2,m)}$ with respect to the $n$-th primitive root of unity $\alpha$.

Notice that the codes $\C_{(0,2,m)}$ and $\C_{(1,3,m)}$ have very poor minimum distances when $m=3$.
These codes may not be interesting.  In this paper, we study $\C_{(0,3,m)}$, $\C_{(1,2,m)}$,
$\C_{(0,1,m)}$ and $\C_{(2,3,m)}$  only for odd $m$,
as the parameters of these four codes for some even $m$ are not the best known.   We will also study the parameters of the duals of these ternary cyclic codes.

\section{The dimensions of the ternary codes $\C_{(i_1,i_2,m)}$} \label{sec-dimension}

In this section, we  always assume that $m$ is odd and we mainly determine the dimensions of the ternary codes $\C_{(0,3,m)}$, $\C_{(1,2,m)}$,
$\C_{(0,1,m)}$ and $\C_{(2,3,m)}$. To this end, we compute the cardinalities of the sets
$$T_i = \{ 1 \leq j \leq n-1: \ \ w_{3}(j) \equiv i\pmod{4}\},$$
where $i \in\{0,1,2,3\}$.

Clearly, we see that
$$T_0  \cup T_2=\{ 1 \leq j \leq n-1: \ \ w_{3}(j) \equiv 0\pmod{2}\},$$
$$T_1 \cup T_3=\{ 1 \leq j \leq n-1: \ \ w_{3}(j) \equiv 1\pmod{2}\},$$ and both are disjoint unions.
It is easy to see that $w_{3}(j) \equiv j\pmod{2}$. Then we have
\begin{equation}\label{equ-sumTi}|T_0 \cup T_2|=|T_0|+|T_2|=(3^m-3)/2,\quad |T_1 \cup T_3|=|T_1|+|T_3|=(3^m-1)/2.\end{equation}

We begin to determine $|T_i|$ for $i\in\{0,1,2,3\}$.
Consider the polynomial $(1+x+x^2)^m \in \mathbb{Z}[x]$.
It is easy to see the coefficient of $x^t$ in $(1+x+x^2)^m$ is equal to
\begin{equation}\label{equ-k1k2}\sum_{\substack{k_1, \ k_2\ge 0 \\ k_1+2k_2=t}}{m\choose k_1}{m-k_1\choose k_2},\end{equation}
 where $t=0,1,\cdots,2m$.

Let $\omega=e^{\frac{2\pi \sqrt{-1}}{4}}\in\mathbb{C}$ be a $4$-th primitive root of unity. Then $\omega^2+1=0$.
Taking $x=\omega$, it then follows from \eqref{equ-k1k2} that
$$
\omega^m=(1+\omega+\omega^2)^m=s_0+s_1\omega+s_2\omega^2+s_3\omega^3=s_0-s_2+(s_1-s_3)\omega,
$$
where $$
s_i=\sum_{
\substack{
0\le t\le 2m\\
t\equiv i\pmod 4
}}
\sum_{\substack{k_1, \ k_2 \ge 0 \\ k_1+2k_2=t}}{m\choose k_1}{m-k_1\choose k_2} 
$$
for $i\in\{0,1,2,3\}$.
 Then we have
 $$s_0=s_2 \ \ \text{ and } \ \ s_1-s_3=\begin{cases} 1 & \text{ if } m \equiv 1 \pmod 4,\\
                                              -1 & \text{ if } m \equiv 3 \pmod 4.\end{cases}$$
 It is easily deduced from the definitions of $T_i$ that
 $$|T_0|=s_0-1=s_2-1=|T_2|, \  |T_1|=s_1, \ |T_3|=s_3.$$
 Then by \eqref{equ-sumTi} we have
\begin{equation}\label{equ-Ti}\begin{cases}|T_0|=|T_1|-1=|T_2|=|T_3|=(3^m-3)/4 & \text{ if } m \equiv 1 \pmod 4,\\
 |T_0|=|T_1|=|T_2|=|T_3|-1=(3^m-3)/4 & \text{ if } m \equiv 3 \pmod 4.\end{cases}\end{equation}

When $m \equiv 3 \pmod 4\ge3$ is odd, the dimensions of $\C_{(i_1,i_2,m)}$ are treated in the following theorem.

\begin{theorem}\label{thm-dim43}
Let $m \equiv 3 \pmod 4\ge3$. Then we have the following.
\begin{itemize}
  \item The ternary cyclic code $\C_{(0,3,m)}$ has length $n=3^m-1$ and
dimension $k = n/2$, and the ternary cyclic code $\C_{(1,2,m)}$ has length $n=3^m-1$ and dimension $k = (n+2)/2$.
  \item The ternary cyclic code $\C_{(2,3,m)}$ has length $n=3^m-1$ and dimension $k = n/2$,
and the ternary cyclic code $\C_{(0,1,m)}$ has length $n=3^m-1$ and dimension $k = (n+2)/2$.
\end{itemize}
\end{theorem}

\begin{proof}
Note that $T_{(i_1,i_2,m)} =T_{i_1} \cup T_{i_2}$ is a disjoint union and
$$\dim(\C_{(i_1,i_2,m)})=n-|T_{(i_1,i_2,m)}|=n-|T_{i_1}|- |T_{i_2}|.$$
The desired conclusions then follow from \eqref{equ-Ti} directly.
\end{proof}

When $m \equiv 1 \pmod 4\ge5$ is odd, the dimensions of $\C_{(i_1,i_2,m)}$ are given in the following theorem 
whose proof is similar to that of Theorem \ref{thm-dim43} and is omitted here. 

\begin{theorem}\label{thm-dim41}
Let $m \equiv 1 \pmod 4\ge5$. Then we have the following. 
\begin{itemize}
  \item The ternary cyclic code $\C_{(0,3,m)}$ has length $n=3^m-1$ and
dimension $k = (n+2)/2$,
and the ternary cyclic code $\C_{(1,2,m)}$ has length $n=3^m-1$ and dimension $k = n/2$.
  \item The ternary cyclic code $\C_{(2,3,m)}$ has length $n=3^m-1$ and dimension $k = (n+2)/2$,
and the ternary cyclic code $\C_{(0,1,m)}$ has length $n=3^m-1$ and dimension $k = n/2$.
\end{itemize}
\end{theorem}

\section{The ternary code $\C_{(0,3,m)}$ } \label{sec-four1}

In this section, we develop a lower bound on minimum distance of the ternary cyclic code $\C_{(0,3,m)}$.
When $m=3$, the parameters of the code $\C_{(0,3,3)}$ are documented in the following example
and are the parameters of the best known linear codes \cite{G}. This code is also an optimal
cyclic code \cite[Table A.92]{DingBK18}. This fact motivates us to investigate the parameters of the code $\C_{(0,3,m)}$ for odd $m$.

\begin{example}
Let $m=3$. Then $\C_{(0,3,m)}$ has parameters $[26, 13, 8]$ and generator polynomial 
$$
x^{13} + 2x^{11} + x^{10} + x^8 + x^6 + x^4 + 2x^3 + 1. 
$$
\end{example}

The following two lemmas on the defining set $T_{(0,3,m)}$ are necessary to derive the lower bound on minimum distance of $\C_{(0,3,m)}$.

\begin{lemma}\label{lem-031}
Let $m \equiv 1 \pmod{4} \ge 5$.  Let $v=(3^{(m+1)/2}-1)/2$ and $\delta=(3^{(m-1)/2}+5)/2$. Then $\gcd(v, n)=1$. Define
$$
T_{(0,3,m)}(v)=\{vi \bmod n: i \in T_{(0,3,m)} \}.
$$
Then
$$
\{1,2, \ldots, \delta-1\} \subset T_{(0,3,m)}(v).
$$
\end{lemma}

\begin{proof}
Since $v=(3^{(m+1)/2}-1)/2$, we have
\begin{align*}
  \gcd(2v,n) & = \gcd(3^{(m+1)/2}-1,3^m-1) \\
   & = 3^{\gcd((m+1)/2,m)}-1=2.
\end{align*}
It then follows that $\gcd(v, n)=1$.

We now prove the second desired conclusion. It can be easily verified that
\begin{align*}
  v^{-1} \bmod{n}=(3^m-1)/2+3^{(m+1)/2}+1.
\end{align*}
Consider now any positive integer $a$ with $a \leq \delta -1 = (3^{(m-1)/2}+3)/2 < 3^{(m-1)/2}-1$.  

When $a$ is even, we have $av^{-1} \bmod{n} = 3^{(m+1)/2} a +a$. It then follows that
\begin{align*}
w_3(a v^{-1} \bmod{n})=2w_3(a).
\end{align*}
Notice that $w_3(a)\equiv a \equiv 0\pmod{2}$. Consequently,
$$
w_3(a v^{-1} \bmod{n})\equiv 0\pmod{4}.
$$

When $a$ is odd and $a \leq \delta -1$, we have $1 \leq a \leq \delta -2=(3^{(m-1)/2} +1)/2$. It is easily verified that
\begin{align}
a v^{-1} \bmod{n} &=
((3^m-1)/2+ (3^{(m+1)/2}+1)a )\bmod{n}\notag	\\
&=(3^{(m-1)/2}+3^{(m+1)/2}\left( \frac{3^{(m-1)/2}-1}2+a \right)+\frac{3^{(m-1)/2}-1}2+a)\bmod{n}.\label{EQ:031}
\end{align}
The rest of the proof can be divided into the following two cases by noting that $a$ is odd.
\begin{enumerate}
\item Let $1\leq a\leq \frac{3^{(m-1)/2}-3}2$ be odd. It follows from (\ref{EQ:031}) that
$$w_3( a v^{-1} \bmod{n} )=1+2w_3( (3^{(m-1)/2}-1)/2+a).$$	
Note that $(3^{(m-1)/2}-1)/2+a$ is odd. Then we have
$$w_3( (3^{(m-1)/2}-1)/2+a)\equiv 1\pmod{2}.$$ It then follows that $w_3( a v^{-1} \bmod{n} )\equiv 3\pmod{4}$.
\item Let $ a= \frac{3^{(m-1)/2}+1}2$. It follows from (\ref{EQ:031}) that
\begin{align*}
w_3( a v^{-1} \bmod{n} )& =w_3((3^{(m-1)/2}+(3^{(m+1)/2}+1) 3^{(m-1)/2})\bmod{n})\\
& =w_3(2\cdot3^{(m-1)/2}+1)\equiv 3\pmod{4}.
\end{align*}
\end{enumerate}

The desired second conclusion then follows.
\end{proof}

\begin{lemma}\label{lem-032}
Let $m \equiv 3 \pmod{4}>3$.  Let $v=(3^{(m+1)/2}+1)/2$ and $\delta=(3^{(m-1)/2}+7)/2$. Then $\gcd(v, n)=1$. Define
$$
T_{(0,3,m)}(v)=\{vi \bmod n: i \in T_{(0,3,m)} \}.
$$
Then
$$
\{1,2, \ldots, \delta-1\} \subset T_{(0,3,m)}(v).
$$
\end{lemma}

\begin{proof}
Since $v=(3^{(m+1)/2}+1)/2$ and $m/\gcd((m+1)/2,m)=1$ is odd, we have
$$\gcd(2v,n) = \gcd(3^{(m+1)/2}+1,3^m-1) =2.$$
It then follows that $\gcd(v, n)=1$.

We now prove the desired second conclusion. It can be easily verified that
\begin{align*}
  v^{-1} \bmod{n}=(3^m-1)/2+3^{(m+1)/2}-1.
\end{align*}
Consider now any positive $a$ with $a \leq \delta -1 = (3^{(m-1)/2}+5)/2 < 3^{(m-1)/2}-1$.  

When $a$ is even, we have
\begin{align*}
av^{-1} \bmod{n} &=3^{(m+1)/2} a -a \\
&=3^{(m+1)/2}(a-1)+3^{(m+1)/2}-a.
\end{align*}
It then follows that
\begin{align*}
w_3(a v^{-1} \bmod{n})
&= w_3(a-1)+w_3(3^{(m+1)/2}-1-(a-1))\\
&=w_3(3^{(m+1)/2}-1)\\
&=m+1.
\end{align*}
Consequently,
$$
w_3(a v^{-1} \bmod{n})\equiv  0 \pmod{4}.
$$

When $a$ is odd and $a \leq \delta -1$, we have $1 \leq a \leq \delta -2=(3^{(m-1)/2} +3)/2$. 
It is easily verified that
\begin{align}
a v^{-1} \bmod{n} &=
((3^m-1)/2+ (3^{(m+1)/2}-1)a ) \bmod{n}\notag	\\
&=(3^{(m-1)/2}+3^{(m+1)/2}\left( \frac{3^{(m-1)/2}-1}2+a \right)+\frac{3^{(m-1)/2}-1}2-a) \bmod{n}.\label{EQ:032}
\end{align}
The rest of the proof can be divided into the following two cases.
\begin{enumerate}
\item Let $1\leq a\leq \frac{3^{(m-1)/2}-1}2$ be odd. It follows from (\ref{EQ:032}) that
\begin{align*}
a v^{-1} \bmod{n} &=3^{(m-1)/2}+3^{(m+1)/2}\left( \frac{3^{(m-1)/2}-1}2+a \right)+3^{(m-1)/2}-1-\left(\frac{3^{(m-1)/2}-1}2+a\right).
\end{align*}
Note that $(3^{(m-1)/2}-1)/2+a\le3^{(m-1)/2}-1$, then we have
\begin{align*}
w_3( a v^{-1} \bmod{n} )&=1+w_3( (3^{(m-1)/2}-1)/2+a)+w_3(3^{(m-1)/2}-1-((3^{(m-1)/2}-1)/2+a))\\
&=1+w_3(3^{(m-1)/2}-1)\\
&=m\equiv 3\pmod{4}.
\end{align*}	
\item Let $ a= \frac{3^{(m-1)/2}+3}2$. It follows from (\ref{EQ:032}) that
\begin{align*}
w_3( a v^{-1} \bmod{n} )&=w_3((3^{(m+1)/2}(3^{(m-1)/2}+1)+3^{(m-1)/2}-2) \bmod{n})\\
&=w_3(3^{(m+1)/2}+3^{(m-1)/2}-1)\\
&=m\equiv 3\pmod{4}.
\end{align*}	
\end{enumerate}
The desired second conclusion then follows.
\end{proof}

The following theorem gives two lower bounds on the minimum distance of $\C_{(0,3,m)}$, which are obtained from Lemmas \ref{lem-031} and \ref{lem-032} directly
by employing the BCH bound on cyclic codes. 
The dimension of the code follows from Theorems \ref{thm-dim43} and  \ref{thm-dim41}.

\begin{theorem}\label{thm03}
Let $m \geq 5$ be odd.
Then the ternary cyclic code $\C_{(0,3,m)}$ has  parameters $[n, k, d]$ with
\begin{eqnarray*}
d  \geq \left\{
\begin{array}{ll}
\frac{3^{(m-1)/2}+7}{2}  & {\rm  if} \  m \equiv 3 \pmod{4}, \\
\frac{3^{(m-1)/2}+5}{2} & {\rm  if} \ m \equiv 1 \pmod{4}
\end{array}
\right.
\end{eqnarray*}
and
\begin{eqnarray*}
k=\left\{
\begin{array}{ll}
\frac{n}{2} & {\rm if~} m \equiv 3 \pmod{4}, \\
\frac{n+2}{2} & {\rm if~} m \equiv 1 \pmod{4}.
\end{array}
\right.
\end{eqnarray*}
\end{theorem}

Let $\delta_{\max}(v)$ denote the largest value $\ell$ such that
$$
\{a, (a +1) \bmod n, \ldots, (a + \ell -1) \bmod n\}  \subset T_{(0,3,m)}(v)
$$
for some integer $a$ with $0 \leq a \leq n-1$. Define
$$
\delta_{\max}=\max\{ \delta_{\max}(v): 1 \leq v \leq n-1, \ \gcd(v,n)=1\}.
$$
Then we have the following results on $m$ and corresponding $\delta_{\max}$ by Magma:
$$(m, \delta_{\max})=(3, 5), \ (5, 11), \ (7, 19), \ (9, 43).$$
When $m=9$, we have
$$
(3^{(m-1)/2}+5)/2  = 43.
$$
Hence, lower bounds in Theorem \ref{thm03}  cannot be improved with the BCH bound on cyclic codes in general.

\section{The ternary cyclic codes $\C_{(1,2,m)}$ } \label{sec-four2}

The parameters of the code $\C_{(1,2,3)}$ are documented in the following example
and are the parameters of the best linear codes known \cite{G}. This code is an optimal
cyclic code \cite[Table A.92]{DingBK18}.

\begin{example}
Let $m=3$. Then $\C_{(1,2,m)}$ has parameters $[26, 14, 7]$ and generator polynomial 
$$
x^{12} + x^{11} + 2x^{10} + x^9 + 2x^8 + 2x^7 + x^6 + x^5 + x^4 + 
    2x^3 + x^2 + x + 1. 
$$
\end{example}

The following two lemmas on the defining set $T_{(1,2,m)}$ are necessary to derive the lower bound on the minimum distance of $\C_{(1,2,m)}$.

\begin{lemma}\label{lem-121}
Let $m \equiv 1 \pmod{4}>1$.  Let $v=(3^{(m-1)/2}+1)/2$ and $\delta=(3^{(m-1)/2}+13)/2$. Then $\gcd(v, n)=1$. Define
$$
T_{(1,2,m)}(v)=\{vi \bmod n: i \in T_{(1,2,m)} \}.
$$
Then
$$
\{1, 2, \ldots, \delta-1\} \subset T_{(1,2,m)}(v).
$$
\end{lemma}

\begin{proof}
Since $v=(3^{(m-1)/2}+1)/2$ and $m/\gcd((m-1)/2,m)=1$ is odd, we have
$$\gcd(2v,n) = \gcd(3^{(m-1)/2}+1,3^m-1) =2.$$
It then follows that $\gcd(v, n)=1$.

We now prove the desired second conclusion. It can be easily verified that
\begin{align*}
  v^{-1} \bmod{n}=(3^m-1)/2-3^{(m+1)/2}+3.
\end{align*}
Consider now any positive $a$ with $a \leq \delta -1 = (3^{(m-1)/2}+11)/2 < 3^{(m-1)/2}-1$.  

When $a$ is even, we have
\begin{align*}
av^{-1} \bmod{n} =n-3^{(m+1)/2} a +3a,
\end{align*}
and
\begin{align*}
3^{m-1}av^{-1} \bmod{n} &=n-3^{(m-1)/2} a +a\\
&=3^{(m-1)/2}(3^{(m+1)/2}-a )+a-1.
\end{align*}
It then follows that
\begin{align*}
w_3(a v^{-1} \bmod{n})&=w_3(3^{m-1} a v^{-1} \bmod{n})\\
&=w_3(a-1)+w_3(3^{(m+1)/2}-1-(a-1) )\\
&=w_3(3^{(m+1)/2}-1)\\
&=m+1.
\end{align*}
Consequently,
$$
w_3(a v^{-1} \bmod{n})\equiv 2 \pmod{4}.
$$

When $a$ is odd and $a \leq \delta -1$, we have $1 \leq a \leq \delta -2=(3^{(m-1)/2} +9)/2$. It is easily verified that
$$a v^{-1} \bmod{n} =((3^m-1)/2- (3^{(m-1)/2}-1)3a)\bmod{n},$$
 and
\begin{align}
3^{m-1} a v^{-1} \bmod{n}&= (3^m-1)/2- 3^{(m-1)/2}a+a \notag	\\
&=3^{m-1}+3^{(m-1)/2}\left( \frac{3^{(m-1)/2}-1}2-a \right)+\frac{3^{(m-1)/2}-1}2+a.\label{EQ:121}
\end{align}
Clearly, $w_3( a v^{-1} \bmod{n} )=w_3(3^{m-1}a v^{-1} \bmod{n})$. The rest of the proof can be divided into the following four cases.
\begin{enumerate}
\item Let $1\leq a\leq \frac{3^{(m-1)/2}-1}2$ be odd. It follows from (\ref{EQ:121}) that 
\begin{align*}
& 3^{m-1} a v^{-1} \bmod{n}  \\ &=3^{m-1}+3^{(m-1)/2}\left(\frac{3^{(m-1)/2}-1}2-a\right)+ 3^{(m-1)/2}-1- \left(\frac{3^{(m-1)/2}-1}2-a\right).
\end{align*}
Then we have
\begin{align*}
&w_3( a v^{-1} \bmod{n} )=w_3( 3^{m-1}a v^{-1} \bmod{n} )\\
=&~1+w_3((3^{(m-1)/2}-1)/2-a)+w_3( 3^{(m-1)/2}-1-( (3^{(m-1)/2}-1)/2-a ) )\\
=&~1+w_3(3^{(m-1)/2}-1)=m\equiv 1\pmod{4}.
\end{align*}	
\item Let $ a= \frac{3^{(m-1)/2}+1}2$. It follows from (\ref{EQ:121}) that
$$w_3( a v^{-1} \bmod{n} )=w_3(3^{m-1})\equiv 1\pmod{4}.$$	
\item Let $ a= \frac{3^{(m-1)/2}+5}2$. It follows from (\ref{EQ:121}) that
$$w_3( a v^{-1} \bmod{n} )=w_3(3^{m-1}-3^{(m+1)/2}+3^{(m-1)/2}+2)=m\equiv 1\pmod{4}.$$	
\item Let $ a= \frac{3^{(m-1)/2}+9}2$. It follows from (\ref{EQ:121}) that
$$w_3( a v^{-1} \bmod{n} )=w_3(3^{m-1}-3^{(m+1)/2}-3^{(m-1)/2}+3+1)=m\equiv 1\pmod{4}.$$	
\end{enumerate}
The desired second conclusion then follows.
\end{proof}

\begin{lemma}\label{lem-122}
Let $m \equiv 3 \pmod{4} > 3$.  Let $v=(3^{(m-1)/2}-1)/2$ and $\delta=(3^{(m-1)/2}+11)/2$. Then $\gcd(v, n)=1$. Define
$$
T_{(1,2,m)}(v)=\{vi \bmod n: i \in T_{(1,2,m)} \}.
$$
Then
$$
\{1,2, \ldots, \delta-1\} \subset T_{(1,2,m)}(v).
$$
\end{lemma}

\begin{proof}
Since $v=(3^{(m-1)/2}-1)/2$, we have
\begin{align*}
  \gcd(2v,n) & = \gcd(3^{(m-1)/2}-1,3^m-1) \\
   & = 3^{\gcd((m-1)/2,m)}-1=2.
\end{align*}
It then follows that $\gcd(v, n)=1$.

We now prove the desired second conclusion. It can be easily verified that
\begin{align*}
  v^{-1} \bmod{n}=(3^m-1)/2-3^{(m+1)/2}-3.
\end{align*}
Consider now any positive $a$ with $a \leq \delta -1 = (3^{(m-1)/2}+9)/2 < 3^{(m-1)/2}-1$. 

When $a$ is even, we have
\begin{align*}
a  v^{-1} \bmod{n} =
n- 3^{(m+1)/2} a -3a, 
\end{align*}
and 
\begin{align*}
3^{m-1}a  v^{-1} \bmod{n} &=
n- 3^{(m-1)/2} a -a\\
&=3^{(m-1)/2}(3^{(m+1)/2}-1-a )+3^{(m-1)/2}-1-a.
\end{align*}
It then follows that
\begin{align*}
w_3(a  v^{-1} \bmod{n})
&= w_3(3^{m-1} a  v^{-1} \bmod{n} )\\
&=w_3(3^{(m+1)/2}-1-a )+w_3(3^{(m-1)/2}-1-a )\\
&=m+1-w_3(a)+m-1-w_3(a)\\
&=2m-2w_3(a).
\end{align*}
Notice that $w_3(a)\equiv a \equiv 0\pmod{2}$. Consequently,
$$
w_3(a  v^{-1} \bmod{n})\equiv  2m \equiv  2 \pmod{4}.
$$ 

When $a$ is odd and $a \leq \delta -1$, we have $1 \leq a \leq \delta -2=(3^{(m-1)/2} +7)/2$. It is easily verified that
$$a v^{-1} \bmod{n} =
(3^m-1)/2- (3^{(m-1)/2}+1)3a\bmod{n},$$
 and
\begin{align}
3^{m-1} a v^{-1} \bmod{n}&=(3^m-1)/2- (3^{(m-1)/2}+1)a \notag	\\
&=3^{m-1}+3^{(m-1)/2}\left( \frac{3^{(m-1)/2}-1}2-a \right)+\frac{3^{(m-1)/2}-1}2-a.\label{EQ:122}
\end{align}
Clearly, $w_3( a v^{-1} \bmod{n} )=w_3(3^{m-1}a v^{-1} \bmod{n})$. The rest of the proof can be divided into the following three cases.
\begin{enumerate}
\item Let $1\leq a\leq \frac{3^{(m-1)/2}-1}2$ be odd. It follows from (\ref{EQ:122}) that
$$w_3( a v^{-1} \bmod{n} )=1+2w_3( (3^{(m-1)/2}-1)/2-a).$$	
Note that $(3^{(m-1)/2}-1)/2-a$ is even, we have
$$w_3( (3^{(m-1)/2}-1)/2-a)\equiv 0\pmod{2}.$$ It then follows that $w_3( a v^{-1} \bmod{n} )\equiv 1\pmod{4}$.
\item Let $ a= \frac{3^{(m-1)/2}+3}2$. It follows from (\ref{EQ:122}) that
$$w_3( a v^{-1} \bmod{n} )=w_3(3^{m-1}-2\cdot3^{(m-1)/2}-2)=2m-5\equiv 1\pmod{4}.$$	
\item Let $ a= \frac{3^{(m-1)/2}+7}2$. It follows from (\ref{EQ:122}) that 
\begin{align*}
w_3( a v^{-1} \bmod{n} ) &=
w_3(3^{m-1}-3^{(m+1)/2}-3^{(m-1)/2}-3-1) \\
&= 2m-5\equiv 1\pmod{4}.
\end{align*} 
\end{enumerate}

The desired second conclusion then follows.
\end{proof}

The following theorem gives two lower bounds on the minimum distance of $\C_{(1,2,m)}$, which are obtained from Lemmas \ref{lem-121} and \ref{lem-122} directly
by employing the BCH bound on cyclic codes. 
The dimension of the code follows from Theorems \ref{thm-dim43} and  \ref{thm-dim41}.

\begin{theorem}\label{thm12}
Let $m \geq 5$ be odd.
Then the ternary cyclic code $\C_{(1,2,m)}$ has  parameters $[n, k, d]$ with
\begin{eqnarray*}
d \geq \left\{
\begin{array}{ll}
(3^{(m-1)/2}+13)/2 & {\rm if~} m \equiv 1 \pmod{4}, \\
(3^{(m-1)/2}+11)/2  & {\rm if ~} m \equiv 3 \pmod{4}
\end{array}
\right.
\end{eqnarray*}
and
\begin{eqnarray*}
k=\left\{
\begin{array}{ll}
\frac{n}{2} & {\rm if~} m \equiv 1 \pmod{4}, \\
\frac{n+2}{2} & {\rm if~} m \equiv 3 \pmod{4}.
\end{array}
\right.
\end{eqnarray*}
\end{theorem}   	
	

\section{The ternary cyclic codes $\C_{(0,1,m)}$ } \label{sec-four3}

The parameters of the code $\C_{(0,1,3)}$ are documented in the following example
and are the parameters of the best linear codes known \cite{G}. This code is an optimal
cyclic code \cite[Table A.92]{DingBK18}.

\begin{example}
Let $m=3$. Then $\C_{(0,1,m)}$ has parameters $[26, 14, 7]$ and generator polynomial 
$$
x^{12} + x^{11} + x^{10} + 2x^9 + x^8 + x^7 + x^6 + 2x^5 + 2x^4 + 
    x^3 + 2x^2 + x + 1.
$$ 

\end{example}

The following two lemmas on the defining set $T_{(0,1,m)}$ are necessary to develop the lower bound on the minimum distance of $\C_{(0,1,m)}$.

\begin{lemma}\label{lem-011}
Let $m \equiv 1 \pmod{4}>1$.  Let $v=(3^{(m-1)/2}+1)/2$ and $\delta=(3^{(m-1)/2}+13)/2$. Then $\gcd(v, n)=1$. Define
$$
T_{(0,1,m)}(v)=\{vi \bmod n: i \in T_{(0,1,m)} \}.
$$
Then
$$
\{n-(\delta-1), \ldots, n-2, n-1\} \subset T_{(0,1,m)}(v).
$$
\end{lemma}

\begin{proof}
Since $v=(3^{(m-1)/2}+1)/2$ and $m/\gcd((m-1)/2,m)=1$ is odd, we have
$$\gcd(2v,n) = \gcd(3^{(m-1)/2}+1,3^m-1) =2.$$
It then follows that $\gcd(v, n)=1$.

We now prove the desired second conclusion. It can be easily verified that
\begin{align*}
  v^{-1} \bmod{n}=(3^m-1)/2-3^{(m+1)/2}+3.
\end{align*}
Consider now any positive $a$ with $a \leq \delta -1 = (3^{(m-1)/2}+11)/2 < 3^{(m-1)/2}-1$. 

When $a$ is even, we have
\begin{align*}
(n-a)v^{-1} \bmod{n} =3^{(m+1)/2} a -3a,
\end{align*}
and 
\begin{align*}
3^{m-1}(n-a)v^{-1} \bmod{n} &=3^{(m-1)/2} a -a\\
&=3^{(m-1)/2} (a-1)+3^{(m-1)/2}-1-(a-1).
\end{align*}
It then follows that
\begin{align*}
w_3((n-a) v^{-1} \bmod{n})&=w_3(3^{m-1}(n-a) v^{-1} \bmod{n})\\
&= w_3(a-1)+w_3( 3^{(m-1)/2}-1-(a-1)) \\
&=w_3(3^{(m-1)/2}-1)=m-1.
\end{align*}
Consequently,
$$
w_3((n-a) v^{-1} \bmod{n})\equiv 0 \pmod{4}.
$$

When $a$ is odd and $a \leq \delta -1$, we have $1 \leq a \leq \delta -2=(3^{(m-1)/2} +9)/2$.
 It is easily verified that
$$(n-a) v^{-1} \bmod{n} =
\left((3^m-1)/2+ (3^{(m-1)/2}-1)3a\right)\bmod{n},$$
 and
\begin{align}
3^{m-1} (n-a) v^{-1} \bmod{n} =&
(3^m-1)/2+ 3^{(m-1)/2}a-a \notag	\\
=~&3^{m-1}+3^{(m-1)/2}\left( \frac{3^{(m-1)/2}-1}2+a \right)+\frac{3^{(m-1)/2}-1}2-a.\label{EQ:011}
\end{align} 
Clearly, $w_3( (n-a) v^{-1} \bmod{n} )=w_3(3^{m-1}(n-a) v^{-1} \bmod{n})$. The rest of the proof can be divided into the following four cases.
\begin{enumerate}
\item Let $1\leq a\leq \frac{3^{(m-1)/2}-1}2$ be odd. It follows from (\ref{EQ:011}) that 
\begin{align*}
&3^{m-1}(n-a) v^{-1} \bmod{n} \\ &=3^{m-1}+3^{(m-1)/2}\left(3^{(m-1)/2}-1-\left(\frac{3^{(m-1)/2}-1}2-a\right)\right)+ \frac{3^{(m-1)/2}-1}2-a.
\end{align*} 
Note that $(3^{(m-1)/2}-1)/2-a\le3^{(m-1)/2}-1$, then we have
\begin{align*}
w_3( a v^{-1} \bmod{n} )&=1+w_3(3^{(m-1)/2}-1-((3^{(m-1)/2}-1)/2-a))+w_3( (3^{(m-1)/2}-1)/2-a)\\
&=1+w_3(3^{(m-1)/2}-1)\\
&=m\equiv 1\pmod{4}.
\end{align*}	
\item Let $ a= \frac{3^{(m-1)/2}+1}2$. It follows from (\ref{EQ:011}) that
$$w_3( (n-a) v^{-1} \bmod{n} )=w_3(2\cdot3^{m-1}-1)=2m-1\equiv 1\pmod{4}.$$	
\item Let $ a= \frac{3^{(m-1)/2}+5}2$. It follows from (\ref{EQ:011}) that
$$w_3( (n-a) v^{-1} \bmod{n} )=w_3(2\cdot3^{m-1}+2\cdot3^{(m-1)/2}-3)=m\equiv 1\pmod{4}.$$	
\item Let $ a= \frac{3^{(m-1)/2}+9}2$. It follows from (\ref{EQ:011}) that
$$w_3( (n-a) v^{-1} \bmod{n} )=w_3(2\cdot3^{m-1}+3^{(m+1)/2}+3^{(m-1)/2}-3-2)=m\equiv 1\pmod{4}.$$	
\end{enumerate}
The desired second conclusion then follows.
\end{proof}

\begin{lemma}\label{lem-012}
Let $m \equiv 3 \pmod{4}>3$.  Let $v=(3^{(m-1)/2}-1)/2$ and $\delta=(3^{(m-1)/2}+11)/2$. Then $\gcd(v, n)=1$. Define
$$
T_{(0,1,m)}(v)=\{vi \bmod n: i \in T_{(0,1,m)} \}.
$$
Then
$$
\{n-(\delta-1), \ldots, n-2, n-1\} \subset T_{(0,1,m)}(v).
$$
\end{lemma}

\begin{proof} 
Since $v=(3^{(m-1)/2}-1)/2$, we have
\begin{align*}
  \gcd(2v,n) & = \gcd(3^{(m-1)/2}-1,3^m-1) \\
   & = 3^{\gcd((m-1)/2,m)}-1\\
   & =2.
\end{align*}
It then follows that $\gcd(v, n)=1$.

We now prove the desired second conclusion. It can be easily verified that
\begin{align*}
  v^{-1} \bmod{n}=(3^m-1)/2-3^{(m+1)/2}-3.
\end{align*}

Consider now any positive $a$ with $a \leq \delta -1 = (3^{(m-1)/2}+9)/2 < 3^{(m-1)/2}-1$. 

When $a$ is even, we have
\begin{align*}
(n-a)  v^{-1} \bmod{n} =
3^{(m+1)/2} a +3a,
\end{align*}
and 
\begin{align*}
3^{m-1}(n-a)  v^{-1} \bmod{n} =
3^{(m-1)/2} a +a.
\end{align*}
It then follows that
$$
w_3((n-a)  v^{-1} \bmod{n})= w_3(3^{m-1}(n-a)  v^{-1} \bmod{n})= 2w_3(a).
$$
Notice that $w_3(a)\equiv a \equiv 0\pmod{2}$. Consequently,
$$
w_3((n-a)  v^{-1} \bmod{n})\equiv 0\pmod{4}.
$$

When $a$ is odd and $a \leq \delta -1$, we have $1 \leq a \leq \delta -2=(3^{(m-1)/2} +7)/2$. It is easily verified that
$$(n-a) v^{-1} \bmod{n} =
((3^m-1)/2+ (3^{(m-1)/2}+1)3a)\bmod{n},$$
 and
\begin{align}
3^{m-1} (n-a) v^{-1} \bmod{n} &=(3^m-1)/2+ (3^{(m-1)/2}+1)a \notag	\\
&=3^{m-1}+3^{(m-1)/2}\left( \frac{3^{(m-1)/2}-1}2+a \right)+\frac{3^{(m-1)/2}-1}2+a.\label{EQ:012}
\end{align} 
Clearly, $w_3( (n-a) v^{-1} \bmod{n} )=w_3(3^{m-1}(n-a) v^{-1} \bmod{n})$. The rest of the proof can be divided into the following three cases.
\begin{enumerate}
\item Let $1\leq a\leq \frac{3^{(m-1)/2}-1}2$ be odd. It follows from (\ref{EQ:012}) that
$$w_3( (n-a) v^{-1} \bmod{n} )=1+2w_3( (3^{(m-1)/2}-1)/2+a).$$	
Note that $(3^{(m-1)/2}-1)/2+a$ is even, we have
$$w_3( (3^{(m-1)/2}-1)/2+a)\equiv 0\pmod{2}.$$ It then follows that $w_3( (n-a) v^{-1} \bmod{n} )\equiv 1\pmod{4}$.
\item Let $ a= \frac{3^{(m-1)/2}+3}2$. It follows from (\ref{EQ:012}) that
$$w_3( (n-a) v^{-1} \bmod{n} )=w_3(3^{m-1}+(3^{(m-1)/2}+1)^2)=5\equiv 1\pmod{4}.$$	
\item Let $ a= \frac{3^{(m-1)/2}+7}2$. It follows from (\ref{EQ:012}) that
\begin{align*}
w_3( (n-a) v^{-1} \bmod{n} ) &=
w_3(3^{m-1}+(3^{(m-1)/2}+1)(3^{(m-1)/2}+3)) \\
&=w_3(2\cdot 3^{m-1}+3^{(m+1)/2}+3^{(m-1)/2}+3 ) \\
&= 5\equiv 1\pmod{4}.
\end{align*}
\end{enumerate}

The desired second conclusion then follows.
\end{proof}

The following theorem gives two lower bounds on the minimum distance of $\C_{(0,1,m)}$, which are obtained from Lemmas \ref{lem-011} and \ref{lem-012} directly
by employing the BCH bound on cyclic codes. 
The dimension of the code follows from Theorems \ref{thm-dim43} and  \ref{thm-dim41}. 
	
\begin{theorem}\label{thm01}
Let $m \geq 3$ be odd.
Then the ternary cyclic code $\C_{(0,1,m)}$ has  parameters $[n, k, d]$ with
\begin{eqnarray*}
d \geq \left\{
\begin{array}{ll}
(3^{(m-1)/2}+13)/2 & {\rm if~} m \equiv 1 \pmod{4}, \\
(3^{(m-1)/2}+11)/2  & {\rm if~} m \equiv 3 \pmod{4}
\end{array}
\right.
\end{eqnarray*}
and
\begin{eqnarray*}
k=\left\{
\begin{array}{ll}
\frac{n}{2} & {\rm if~} m \equiv 1 \pmod{4}, \\
\frac{n+2}{2} & {\rm if~} m \equiv 3 \pmod{4}.
\end{array}
\right.
\end{eqnarray*}
\end{theorem}

\section{The ternary cyclic codes $\C_{(2,3,m)}$ } \label{sec-four4}

The parameters of the code $\C_{(2,3,3)}$  are documented in the following example
and are the parameters of the best linear codes known \cite{G}. This code is an optimal
cyclic code \cite[Table A.92]{DingBK18}.

\begin{example}
Let $m=3$. Then $\C_{(2,3,m)}$ has parameters $[26, 13, 8]$ and generator polynomial 
$$ 
x^{13} + 2x^{10} + x^9 + x^7 + x^5 + x^3 + 2x^2 + 1. 
$$ 
\end{example}

The following two lemmas on the defining set $T_{(2,3,m)}$ are necessary to develop the lower bound on 
the minimum distance of $\C_{(2,3,m)}$.
Their proofs are very similar to those of Lemmas \ref{lem-011} and \ref{lem-012} and omitted here.

\begin{lemma}\label{lem-231}
Let $m \equiv 1 \pmod{4}>1$.  Let $v=(3^{(m+1)/2}-1)/2$ and $\delta=(3^{(m-1)/2}+5)/2$. Then $\gcd(v, n)=1$. Define
$$
T_{(2,3,m)}(v)=\{vi \bmod n: i \in T_{(2,3,m)} \}.
$$
Then
$$
\{n-(\delta-1), \ldots, n-2, n-1\} \subset T_{(2,3,m)}(v).
$$
\end{lemma}

\begin{lemma}\label{lem-232}
Let $m \equiv 3 \pmod{4}>3$.  Let $v=(3^{(m+1)/2}+1)/2$ and $\delta=(3^{(m-1)/2}+7)/2$. Then $\gcd(v, n)=1$. Define
$$
T_{(2,3,m)}(v)=\{vi \bmod n: i \in T_{(2,3,m)} \}.
$$
Then
$$
\{n-(\delta-1), \ldots, n-2, n-1\} \subset T_{(2,3,m)}(v).
$$
\end{lemma}

The following theorem gives two lower bounds on the minimum distance of $\C_{(2,3,m)}$, which are obtained from Lemmas \ref{lem-231} and \ref{lem-232} directly
by employing the BCH bound on cyclic codes. 
The dimension of the code follows from Theorems \ref{thm-dim43} and  \ref{thm-dim41}. 
	
\begin{theorem}\label{thm23}
Let $m \geq 3$ be odd.
Then the ternary cyclic code $\C_{(2,3,m)}$ has  parameters $[n, k, d]$ with
\begin{eqnarray*}
d  \geq \left\{
\begin{array}{ll}
\frac{3^{(m-1)/2}+7}{2}  & {\rm if ~} m \equiv 3 \pmod{4}, \\
\frac{3^{(m-1)/2}+5}{2} & {\rm if~} m \equiv 1 \pmod{4}
\end{array}
\right.
\end{eqnarray*}
and
\begin{eqnarray*}
k=\left\{
\begin{array}{ll}
\frac{n}{2} & {\rm if~} m \equiv 3 \pmod{4}, \\
\frac{n+2}{2} & {\rm if~} m \equiv 1 \pmod{4}.
\end{array}
\right.
\end{eqnarray*}
\end{theorem}

\section{The dual codes of these ternary cyclic codes } \label{sec-dual}

In this section, we investigate the parameters of the dual codes of $\C_{(0,3,m)}$, $\C_{(1,2,m)}$,
$\C_{(0,1,m)}$ and $\C_{(2,3,m)}$ for odd $m$.

\subsection{The ternary cyclic codes $\C_{(0,3,m)}^\perp$}

The parameters of the code $\C_{(0,3,3)}^\perp$ are documented in the following example
 and are the parameters of the best linear codes known \cite{G}. This code is an optimal
cyclic code \cite[Table A.92]{DingBK18}.

\begin{example}
Let $m=3$. Then
the dual code $\C_{(0,3,m)}^\perp$ has parameters $[26, 13, 8]$ and generator polynomial 
$$
x^{13} + x^{10} + 2x^9 + x^6 + 2x^4 + x^3 + 2x^2 + 2. 
$$ 
\end{example}

\begin{theorem}\label{thm03dual}
Let $m \geq 3$ be odd.
Then the ternary cyclic code $\C_{(0,3,m)}^\perp$ has  parameters $[n, k, d]$ with
\begin{eqnarray*}
d \geq \left\{
\begin{array}{ll}
(3^{(m-1)/2}+15)/2 & {\rm if~} m \equiv 1 \pmod{4}, \\
(3^{(m-1)/2}+13)/2  & {\rm if~} m \equiv 3 \pmod{4}
\end{array}
\right.
\end{eqnarray*}
and
\begin{eqnarray*}
k=\left\{
\begin{array}{ll}
\frac{n-2}{2} & {\rm if~} m \equiv 1 \pmod{4}, \\
\frac{n}{2} & {\rm if~} m \equiv 3 \pmod{4}.
\end{array}
\right.
\end{eqnarray*}
\end{theorem}

\begin{proof}
Note that
$$
(x-1)g_{(0,3,m)}(x)g_{(1,2,m)}(x)=x^n-1.
$$
The complement code $\C_{(0,3,m)}^c$ of $\C_{(0,3,m)}$ has generator polynomial
$(x-1)g_{(1,2,m)}(x)$ and is a subcode of $\C_{(1,2,m)}$ with
$$
\dim(\C_{(0,3,m)}^c) = \dim(\C_{(1,2,m)})-1.
$$
It is well-known that $\C_{(0,3,m)}^c$ and $\C_{(0,3,m)}^\perp$ have the same parameters \cite{HP03}. Note that $0$ is contained in the defining set of $\C_{(0,3,m)}^c$. The
desired conclusions then follow from Theorem \ref{thm12}.
\end{proof}

\subsection{The ternary cyclic codes $\C_{(1,2,m)}^\perp$ }

The parameters of the code $\C_{(1,2,3)}^\perp$  are documented in the following example
and are the parameters of the best linear codes known \cite{G}. This code is an optimal
cyclic code \cite[Table A.92]{DingBK18}.

\begin{example}
Let $m=3$. Then
the dual code $\C_{(1,2,m)}^\perp$ has parameters $[26, 12, 9]$ and generator polynomial 
$$ 
x^{14} + 2x^{13} + 2x^{11} + 2x^{10} + 2x^9 + x^8 + 2x^7 + x^6 + 2x^5 + x^4 + x^3 
    + x^2 + x + 2.
$$ 
\end{example}

\begin{theorem}\label{thm12dual}
Let $m \geq 3$ be odd.
Then the ternary cyclic code $\C_{(1,2,m)}^\perp$ has  parameters $[n, k, d]$ with
\begin{eqnarray*}
d  \geq \left\{
\begin{array}{ll}
\frac{3^{(m-1)/2}+9}{2}  & {\rm if~} m \equiv 3 \pmod{4}, \\
\frac{3^{(m-1)/2}+7}{2} & {\rm if~} m \equiv 1 \pmod{4}
\end{array}
\right.
\end{eqnarray*}
and
\begin{eqnarray*}
k=\left\{
\begin{array}{ll}
\frac{n-2}{2} & {\rm if~} m \equiv 3 \pmod{4}, \\
\frac{n}{2} & {\rm if~} m \equiv 1 \pmod{4}.
\end{array}
\right.
\end{eqnarray*}
\end{theorem}

\begin{proof}
Note that
$$
(x-1)g_{(0,3,m)}(x)g_{(1,2,m)}(x)=x^n-1.
$$
The complement code $\C_{(1,2,m)}^c$ of $\C_{(1,2,m)}$ has generator polynomial
$(x-1)g_{(0,3,m)}(x)$ and is a subcode of $\C_{(0,3,m)}$. The
desired conclusions then similarly follow from Theorem \ref{thm03}.
\end{proof}

\subsection{The ternary cyclic codes $\C_{(0,1,m)}^\perp$ }

The parameters of the code $\C_{(0,1,3)}^\perp$ are documented in the following example
and are the parameters of the best linear codes known \cite{G}. This code is an optimal
cyclic code \cite[Table A.92]{DingBK18}.

\begin{example}
Let $m=3$. Then
the dual code $\C_{(0,1,m)}^\perp$ has parameters $[26, 12, 9]$ and generator polynomial 
$$ 
x^{14} + 2x^{13} + 2x^{12} + 2x^{11} + 2x^{10} + x^9 + 2x^8 + x^7 + 2x^6 + x^5 + x^4
    + x^3 + x + 2. 
$$ 
\end{example}

\begin{theorem}\label{thm01dual}
Let $m \geq 3$ be odd.
Then the ternary cyclic code $\C_{(0,1,m)}^\perp$ has  parameters $[n, k, d]$ with
\begin{eqnarray*}
d  \geq \left\{
\begin{array}{ll}
\frac{3^{(m-1)/2}+9}{2}  & {\rm if~} m \equiv 3 \pmod{4}, \\
\frac{3^{(m-1)/2}+7}{2} & {\rm if~} m \equiv 1 \pmod{4}
\end{array}
\right.
\end{eqnarray*}
and
\begin{eqnarray*}
k=\left\{
\begin{array}{ll}
\frac{n-2}{2} & {\rm if~} m \equiv 3 \pmod{4}, \\
\frac{n}{2} & {\rm if~} m \equiv 1 \pmod{4}.
\end{array}
\right.
\end{eqnarray*}
\end{theorem}

\begin{proof}
The proof is similar to that of Theorem \ref{thm03dual} and omitted here.
\end{proof}

\subsection{The ternary cyclic codes $\C_{(2,3,m)}^\perp$ }

The parameters of the code $\C_{(2,3,3)}^\perp$ are documented in the following example
and are the parameters of the best linear codes known \cite{G}. This code is an optimal
cyclic code \cite[Table A.92]{DingBK18}.

\begin{example}
Let $m=3$. Then
the dual code $\C_{(2,3,m)}^\perp$ has parameters $[26, 13, 8]$ and generator polynomial 
$$
x^{13} + x^{11} + 2x^{10} + x^9 + 2x^7 + x^4 + 2x^3 + 2. 
$$ 
\end{example}

\begin{theorem}\label{thm23dual}
Let $m \geq 3$ be odd.
Then the ternary cyclic code $\C_{(2,3,m)}^\perp$ has  parameters $[n, k, d]$ with
\begin{eqnarray*}
d \geq \left\{
\begin{array}{ll}
(3^{(m-1)/2}+15)/2 & {\rm if~} m \equiv 1 \pmod{4}, \\
(3^{(m-1)/2}+13)/2  & {\rm if~} m \equiv 3 \pmod{4}
\end{array}
\right.
\end{eqnarray*}
and
\begin{eqnarray*}
k=\left\{
\begin{array}{ll}
\frac{n-2}{2} & {\rm if~} m \equiv 1 \pmod{4}, \\
\frac{n}{2} & {\rm if~} m \equiv 3 \pmod{4}.
\end{array}
\right.
\end{eqnarray*}
\end{theorem}

\begin{proof}
The proof is similar to that of Theorem \ref{thm03dual} and omitted here.
\end{proof}

\section{Summary and concluding remarks} \label{sec-conclude}

The main contributions of this paper are the construction and analysis of the four infinite families of ternary cyclic codes with
parameters $[3^m-1, k \in \{(3^m-1)/2, (3^m+1)/2\}, d]$ for odd $m$, where $d$ has a square-root-like lower bound. Their dual codes have parameters $[3^m-1, k^\perp \in \{(3^m-1)/2, (3^m-3)/2\}, d^\perp]$ for odd $m$, where $d^\perp$ also has a square-root-like lower bound. When $m=3$, these codes are among the best linear codes known and are optimal ternary cyclic codes. The authors are not aware of any ternary cyclic code that is better than a code in the four families of codes and their duals of length $3^m-1$ for odd $m$. These ternary codes presented in this paper are interesting in the sense they are the first infinite families of ternary codes codes with parameters of the form 
$[n, k \in \{n/2, (n \pm 2)/2\}, d]$, where $d$ has a square-root-like lower bound. 

The binary codes in \cite{LiuLiDing23} and the ternary codes $\C_{(i_1, i_2,m)}$ treated in this paper can certainly be generalised to cyclic codes $\C_{(i_1, i_2,q, m)}$
over $\mathbb{F}_q$ for any prime power $q$. But it is amazing that these codes $\C_{(i_1, i_2,q, m)}$
over $\mathbb{F}_q$ may have very
bad parameters for other $q$. For example, we have the following numerical results.
\begin{itemize}
\item $\C_{(2, 3, 5, 3)}$ has parameters $[124, 62, 3]$ and $\C_{(2, 3, 5, 3)}^\perp$ has parameters $[124, 62, 3]$.
\item $\C_{(0, 1, 5, 3)}$ has parameters $[124, 63, 3]$ and $\C_{(0, 1, 5, 3)}^\perp$ has parameters $[124, 61, 4]$.
\item $\C_{(0, 3, 5, 3)}$ has parameters $[124, 63, 3]$ and $\C_{(0, 3, 5, 3)}^\perp$ has parameters $[124, 61, 4]$.
\item $\C_{(1, 2, 5, 3)}$ has parameters $[124, 62, 3]$ and $\C_{(1, 2, 5, 3)}^\perp$ has parameters $[124, 62, 3]$.
\item $\C_{(1, 3, 5, 3)}$ has parameters $[124, 62, 2]$ and $\C_{(1, 3, 5, 3)}^\perp$ has parameters $[124, 62, 2]$.
\item $\C_{(0, 2, 5, 3)}$ has parameters $[124, 63, 2]$ and $\C_{(0, 2, 5, 3)}^\perp$ has parameters $[124, 61, 4]$.
\end{itemize}
In addition, the development of the lower bounds on the minimum distances of these
codes depends on the specific value of $q$. The authors do not see a uniform way for developing the lower
bounds on the binary cyclic codes in \cite{LiuLiDing23} and the lower bounds on the ternary cyclic codes in
this paper. This explains why only the ternary case was treated in this paper and the binary case was considered in
\cite{LiuLiDing23} separately.

\end{document}